\documentclass[12pt,notitlepage,hidelinks,tighten]{revtex4-2}

\usepackage[utf8]{inputenc}
\usepackage{dutchcal}
\usepackage{amsmath}
\usepackage[scr=rsfs]{mathalpha}
\usepackage{amsthm}
\usepackage{amsfonts}
\usepackage{amscd}
\usepackage{comment}
\usepackage{bm}
\usepackage{cancel}
\usepackage{epsfig}
\usepackage{amssymb}
\usepackage{tabularx}
\usepackage{calligra}
\usepackage{enumerate}
\usepackage{hyperref}
\usepackage{tikz}
\usepackage{tikzsymbols}
\usepackage{orcidlink}

\def\a{\alpha}
\def\b{\beta}
\def\g{\gamma}

\def\={\overset{\bm .}{=}}

\def\p{\partial}

\usepackage[T1]{fontenc}
\usepackage[utf8]{inputenc}
\usepackage{tcolorbox}
\usepackage{dutchcal}
\usepackage{enumerate}
\usepackage{tcolorbox}
\tcbuselibrary{breakable}
\usepackage{fontawesome}
\usepackage{lipsum}
\usepackage{bm}
\usepackage{hyperref}
\usepackage{amssymb}
\usepackage{amsmath}
\usepackage{amsfonts}
\usepackage{amscd}
\usepackage{epsfig}
\usepackage{tabularx}
\usepackage{graphicx}

\def\a{\alpha}
\def\b{\beta}
\def\g{\gamma}

\def\d{\delta}

\def\p{\partial}

\usepackage{stackengine}
\stackMath

\newtheorem{thm}{Theorem}
\newtheorem{prop}[thm]{Proposition}
\newtheorem{dfn}[thm]{Definition}

\newtheorem{remark}[thm]{Remark}

\begin{document}

\title{Curvature conditions for generalized singularity 
theorems}

\author{Jeremías Daín$^1$ and Gustavo Dotti$^{1,2}$}

\affiliation{${}^1$FAMAF,  
Universidad Nacional de C\' ordoba  
and ${}^2$IFEG, CONICET.\\ Ciudad Universitaria, (5000) C\'ordoba, Argentina.}
\email{gustavo.dotti@unc.edu.ar}

\begin{abstract}
We study the curvature conditions 
introduced in [Class. Quant. Grav. \textbf{27}, 152002]  to 
predict focal points for trapped 
spacelike submanifolds 
in spacetimes of 
arbitrary dimensions, with the purpose of  
generalizing Penrose's singularity theorem 
to compact trapped submanifolds (CTMs) of codimension  
higher than two. 
We find that these conditions do not 
apply in general but may apply for specific CTMs. As a result, higher codimension
CTMs may still work 
as singularity predictors, although the possibility that they 
intersect the domain of outer communications cannot be ruled 
out 
using standard arguments. 
\end{abstract}

\maketitle

\tableofcontents{}

\section{Introduction}

Energy conditions in $3+1$ General Relativity are largely motivated 
as hypothesis of singularity theorems and obstruction theorems 
for compact trapped surfaces (CTSs) 
intersecting the domain of outer 
communications (DOC). 
Many 
of these  conditions are satisfied by what we  
consider ordinary matter and fields,  giving these results a relatively 
solid ground (see, however, 
\cite{Barcelo,curiel}). The  cosmological and  
black hole singularity theorems 
are based on the existence of focal points along normal causal 
geodesics to (respectively) 
codimension $1$ and $2$ spacelike submanifolds. 
The generalization of this result to a $d$ dimensional, time oriented 
spacetime $(M,g_{ab})$  containing a $k$ dimensional spacelike submanifold 
$S$ of codimension $d-k \geq 1$ is given as  
 Proposition $1$ in \cite{Gallo} (some related 
 results can be found  in \cite{oniell}, pages 288-293): 
\begin{prop} \label{focal} \textbf{\rm \cite{Gallo}} 
Let $n^a$ be a future causal vector normal to 
a spacelike submanifold 
$S \subset M$ at $p$, $N^a$ the tangent vector of
the geodesic $\g$  it generates, $u$ its affine parameter with 
$u=0$ at $S$, 
$e_A^a$, $A=1,2,...,k$ a  basis of 
$T_pS$, $E_A^a$ the parallel transport of $e_A^a$ along $\g$, 
$h_{AB} = g_{ab}e_A^a \cdot e_B^b = E_A \cdot E_B$ and $h^{ab} = h^{AB} E_A^a E_B^b$. 
Let $H^a$ be 
the mean curvature vector at $p$ of $S \subset M$. 
Assume the 
expansion 
$\theta(n)=H^a n_a =-kc <0$, then, if 
\begin{equation} \label{cc}
R_{abcd} N^a N^c h^{bd} \geq 0 \; \text{ along } \gamma,
\end{equation}
there is a focal point to $S$ on $\g$ at or before $u=1/c$ 
assuming that  $\g$ is defined up to that 
value of $u$. 
\label{prop gallo}
\end{prop}

In this paper we adopt the conventions and definitions  in \cite{Dotti2}. The spacetime $(M,g_{ab})$ is time orientable and 
has dimension $d \geq 3$. 
A compact trapped submanifold (CTM) of $M$ is a proper, 
boundary-less submanifold,  that is compact and has a future 
timelike mean curvature vector field $H$ (this condition is usually 
called \emph{future} trapping in the literature, with an analogous 
past trapping condition for submanifolds with a past timelike $H$), 
a compact trapped surface (CTS) is a codimension two CTM and a 
trapped loop (TL) is a one dimensional CTM. 
Whenever we mention a \textit{future null geodesics orthogonal 
to a spacelike subspace  $W \subset T_pM$} we mean 
the half of the geodesic to the future of $p$. 
We warn the reader that the 
the  sign and 
normalization conventions of the mean curvature vector field is not uniform 
in the literature. \\

\noindent
Proposition \ref{prop gallo} 
is used in \cite{Gallo} to prove the following:

\begin{thm} \label{sing-thm} If a $d-$dimensional spacetime $(M,g_{ab})$ contains a non-compact Cauchy hypersurface and a $k-$CTM of codimension $d-k \geq 2$, and if condition \eqref{cc} 
holds along every future-directed null geodesic orthogonal to $S$, then 
$(M,g_{ab})$ is future null geodesically
incomplete.
\end{thm}

\begin{remark} Note the following  differences 
in the hypothesis of the previous theorem and proposition: 
in the theorem the submanifold must be of codimension $\geq 
2$ (to allow null normals) and compact, and  attention is restricted to its
future normal \textbf{null} geodesics.
\end{remark}

Let us analyze \eqref{cc} assuming 
the $d-$dimensional 
Einstein's equations hold: $R_{ab} \propto T_{ab} - 
\tfrac{1}{d-2} T g_{ab}$. 
In the codimension one case (relevant to  cosmological 
contexts), $n^a$  gives  the future 
timelike direction orthogonal to $S$ and the inequality in 
\eqref{cc} reduces to 
$R_{ab}N^aN^b \geq 0$ along $\gamma$, which is guaranteed 
if the spacetime satisfies the strong energy 
condition (SEC)
\begin{equation} \label{sec}
T_{ab} N^a N^b \geq -\tfrac{1}{2}T, \;\; T =T_{ab} g^{ab} 
\;\; \text{ for timelike } N^a. 
\end{equation}
In the codimension two case with $n^a$ null 
(the case relevant to the singularity Theorem \ref{sing-thm})
we can  see (refer to section \ref{general}, 
also \cite{Gallo})   
that 
\eqref{cc} reduces to 
\begin{equation} \label{nec}
R_{ab} N^a N^b  \geq 0 \; \text{ along } \gamma,
\end{equation}
which is guaranteed if the null energy condition (NEC) 
\begin{equation} \label{NEC}
T_{ab} N^a N^b \geq 0 \;\; \text{ for null } N^a
\end{equation}
holds. Since these energy 
conditions are regarded as natural in 
$3+1$ GR, the singularity theorems are usually stated 
assuming these  hypothesis. A CTS 
 for which 
 condition \eqref{nec} is satisfied \emph{only along 
 its orthogonal null 
 geodesics} is a predictor of null geodesic incompleteness,  yet, for the sake of simplicity,   Penrose's 
 theorem is usually stated with the stronger requirement that 
the NEC \eqref{NEC} holds, i.e., that $R_{ab} N^a N^b  \geq 0$
  \emph{at every spacetime point $p$ and 
for every null vector in $T_pM$}.
Unless we wanted to prove the existence of conjugate points for future null geodesics normal to 
\emph{any} CTS, as in the proof that 
 a CTS cannot intersect the domain of outer communications 
 (DOC, see, e.g., Proposition 12.2.2 in \cite{wald}), requiring the NEC on $(M,g_{ab})$ 
is far more than what is needed. To make this distinction clear for higher codimension 
we introduce the condition analogue of the NEC that assures \eqref{cc} will hold for any $k-$CTM:
\begin{dfn} \label{Dgcc} A $d$ dimensional spacetime $(M,g_{ab})$ satisfies the 
generic $k-$curvature condition (k-GCC) if, at every $p \in M$, for 
any null vector $N^a \in T_pM$ and 
any $k-$dimensional spacelike subspace $W \subset T_pM$ 
orthogonal to $N^a$, 
the condition
\begin{equation}\label{gcc}
R_{abcd}N^a N^c h_W^{bd} \geq 0
\end{equation}
holds, where  ${h_W}_{ab}$ is the induced metric on $W$.
\end{dfn}

Note that, since $W$  is spacelike and normal to a null vector, the possible values for $k$ in this definition are $k=1,2,..., d-2$. In view of our comments above, 
for $k=d-2$ 
this condition is equivalent to the NEC (see next section). 
The generic $k-$curvature condition is stronger 
than required in the singularity Theorem  \ref{sing-thm}, 
as it assures that \eqref{cc} will hold along \emph{any} null 
geodesic, 
and it can be used together with Proposition \ref{focal} 
to generalize to higher dimensional spacetimes and 
higher codimension CTMs the standard proof 
(see, e.g., Proposition 12.2.2. in \cite{wald}, 
Proposition 9.2.1 in \cite{hawk}) that a CTS in an   asymptotically simple spacetime 
cannot  
intersect the DOC. \\

 For  codimension higher than 
 two, condition \eqref{gcc} and its weaker version 
  \eqref{cc} for null $N^a$, 
involve the Weyl piece of the Riemann tensor besides its 
Ricci part. It is then 
not adequate to refer to these as ``energy conditions''; 
for spacetimes satisfying Einstein's equations, 
the naturalness of 
\eqref{gcc} could not be justified by the type of matter and fields that we allow. 
One of the  
purposes of this paper is to study how stringent the 
$k-$GCC introduced in Definition \ref{Dgcc}
are for CTMs of codimension $d-k > 2$. 
This is explored in Section \ref{general}.
The other purpose is to exhibit cases where Theorem
\ref{sing-thm} applies: we  find spacetimes
for which the $k-$GCC does not hold but, still, properly oriented $k-$CTMs can be found that  fulfill the hypothesis 
of this theorem. This is 
 explored in section \ref{specific}, with an emphasis 
 on warped spacetimes and a number of applications on 
  spherically symmetric, $d \geq 4$  
 spacetimes.

\section{The curvature condition as a generic requirement} \label{general}

In this section we analyze the $k-$GCC introduced in 
Definition  \ref{Dgcc}. Our results are gathered in 
Proposition \ref{Pgcc} below. \\

Let $N^a$ be a null vector at $T_pM$ 
and $V$ the $d-1$ dimensional subspace of $T_pM$ 
vectors orthogonal to $N^a$.  
The set of  $k-$dimensional
subspaces $W \subset V$ with an \emph{spacelike} induced 
metric ${h_W}_{ab}$
is an open subset 
$\widetilde{Gr}(k,V)$ of the 
Grassmanian manifold ${Gr}(k,V)$ \cite{Dotti2}.  
Motivated by Definition \ref{Dgcc}, we are interested in 
determining if 
the minimum of 
the  function 
\begin{equation}\label{1}
\widetilde{Gr}(k,V) \ni W \to R_{abcd}N^aN^c h_W^{bd} 
\in \mathbb{R} 
\end{equation}
is nonnegative, so that \eqref{gcc} 
will hold at $p$ for the specific null vector $N^a$. 
To avoid the ambiguity of the scaling $N^a \to x N^a$, under which 
\eqref{1} scales as $x^2$, we will use the fact that the 
spacetime is time oriented to pick a unit future timelike 
vector field $T^a$, and assume from now on that  $N^a$ 
is normalized as 
$N_a T^a=-1$. 

We proceed as   in \cite{Dotti2},  
where a similar problem was analyzed (see page 14 on): 
take any subspace $V_o$ of $V$ such that 
\begin{equation} \label{ds}
V = V_o \oplus \mathbb{R} N
\end{equation}
and let $\pi: V \to V_o$ be the projection 
associated to \eqref{ds}.  
Introduce the Weyl tensor $C_{abcd}$, 
\begin{equation}\label{weyl}
    R_{abcd}=C_{abcd}+\frac{2}{d-2}\left(g_{a[c}R_{d]b} - g_{b[c}R_{d]a} \right) - \frac{2}{(d-1)(d-2)}\left(Rg_{a[c}g_{d]b} \right)
\end{equation}
and define the symmetric bilinear  operators 
\begin{equation} \label{bop0}
\begin{split}
V \otimes V \ni (X,Y) \to C^{(N)}(X,Y)&:=C_{abcd}N^a N^c X^b Y^d\\
V \otimes V \ni(X,Y) \to R^{(N)}(X,Y)&:=R_{abcd}N^a N^c X^b Y^d.
\end{split}
\end{equation} 
In view of the null character of $N$, $\pi$ 
is an isometry: $\pi(X) \cdot \pi(Y) = X \cdot Y$. Also
\begin{equation} \label{bop}
\begin{split}
C^{(N)}(X,Y)&=C^{(N)}(\pi(X), \pi(Y))\\
R^{(N)}(X,Y)&=R^{(N)}(\pi(X), \pi(Y)).
\end{split}
\end{equation} 
The above equalities are 
  easily proved if we use the symmetries of the Weyl and 
  Riemann tensor and 
 a frame
$\{ e_1, e_2,...,e_{d-2},N,L\}$ 
for which  $V_o = \text{span}\{ e_1, e_2,..., e_{d-2} \}$ 
and the only nontrivial 
inner products are $e_i \cdot e_j=\d_{ij}$ and $N \cdot 
L=-1$,
so that 
\begin{equation}\label{im}
g^{ab} =-N^a L^b-L^a N^b + \sum_{i,j=1}^{d-2}\d^{ij} e_i^a e_j^b
\end{equation}
and 
\begin{equation}
\pi(X^N N + X^j e_j)= X^j e_j.
\end{equation}
In view of \eqref{bop}, 
the function defined by \eqref{1}, 
satisfies  
\begin{equation}
R_{abcd}N^aN^c h_W^{bd} =: {\rm tr}_{W} R^{(N)} =
{\rm tr}_{\hat W} R^{(N)},
\end{equation}
where $\hat W := \pi(W)$. 
Choose the above basis such that $\hat W = \text{span}\{e_1,
e_2,...,e_k\}$. Using \eqref{weyl}  we find
\begin{equation}
R^{(N)}(e_i,e_j)= C^{(N)}(e_i,e_j)+ 
\frac{1}{d-2}R_{ac}N^a N^c 
\d_{ij}.
\end{equation}
Taking the trace over $\hat W$ (that is, applying 
$\d^{ij}$ to the above equation and summing over $i,j$ from 
$1$ to $k$) yields
\begin{align}\nonumber
 {\rm tr}_W R^{(N)} &=
{\rm tr}_{\hat W} R^{(N)} \\ &=  
{\rm tr}_{\hat W} C^{(N)}  + 
\frac{k}{d-2}R_{ac}N^aN^c. \label{main}
\end{align}
Note that the second term above is independent of $W \in 
\widetilde{Gr}(k,V)$ and that  for $k=d-2$, 
in view of \eqref{im} and the trace free character of the Weyl 
tensor, the first term 
 vanishes:
\begin{equation}\label{weyl0}
 C_{abcd}N^a N^c 
\sum_{i,j=1}^{d-2}e_i^be_j^d \d^{ij}
= C_{abcd}N^a N^c 
\left( g^{bd}+N^b L^d+L^bN^d \right)=0.
\end{equation}
It is important to keep in mind 
that in  \eqref{main} we deal only with subspaces 
 $\hat W \subset V_o$. Since the restrictions of the maps 
 $C^{(N)}$ and $R^{(N)}$ to $V_o \otimes V_o$ are 
 symmetric, and 
 the induced metric $h_{V_o}$ is positive definite,  
the commuting linear operators $V_o \to V_o$ defined by 
$h_{V_o}^{-1} [C^{(N)}\mid_{V_o \otimes V_o}]$ and 
$h_{V_o}^{-1} [R^{(N)}\mid_{V_o \otimes V_o}] 
=h_{V_o}^{-1} [C^{(N)}\mid_{V_o \otimes V_o}] 
+ \frac{1}{d-2}R_{ac}N^aN^c \bm{I}$ 
are simultaneously diagonalizable with spectra 
\begin{equation}\label{evs} 
\begin{split}
&\lambda_1 \leq \lambda_2 \leq ... \leq \lambda_{d-2} \\
&\mu_1 \leq \mu_2 \leq ... \leq \mu_{d-2}
\end{split}
\end{equation}
respectively, where 
\begin{equation}
\mu_k = \lambda_k + \frac{1}{d-2}R_{ac}N^aN^c,
\end{equation}
and, in view of \eqref{weyl0},
\begin{equation}\label{zero}
\sum_{\alpha=1}^{d-2} \lambda_\alpha 
= {\rm tr}_{V_o} C^{(N)} = 0.
\end{equation}

As in \cite{Dotti2} we can prove that the trace 
function \eqref{1} effectively reaches 
a minimum on the open manifold $\widetilde{Gr}(k,V)$ 
and that this minimum is (note the ordering of the 
eigenvalues in \eqref{evs})
\begin{equation}\label{min}
\left.\text{min} \left({\rm tr}_W R^{(N)}\right) \right|_{\widetilde{Gr}(k,V)} 
=\sum_{\alpha=1}^k \mu_\alpha=\sum_{\alpha=1}^k 
\lambda_\alpha + \frac{k}{d-2}R_{ac}N^aN^c.
\end{equation}

\begin{prop} \label{Pgcc}
The $k-$GCC has the following properties:
\begin{enumerate}[i)]
\item For  $k=d-2$   is equivalent to the NEC.
\item For conformally flat spacetimes 
is equivalent to the NEC.
\item  If it holds for $k$ then it holds for  $k'>k$.
\item For a non-flat vacuum spacetime  it is only 
satisfied for $k = d-2$. 
\item For  $k < d-2$ the first term in the second
equality in \eqref{min} 
gives a 
nonpositive contribution.
\end{enumerate}
\end{prop}

\begin{proof} \leavevmode
\begin{enumerate}[i)]
\item  This follows from \eqref{zero} and \eqref{min}.
\item  For conformally flat spacetimes $C^{(N)}$ is trivial,
all the $\lambda's$ are zero and \eqref{min} 
reduces to $\frac{k}{d-2}R_{ac}N^aN^c$.
\item From the first equality \eqref{min} 
we infer that $\sum_{\alpha=1}^k \mu_\alpha \geq 0$. In 
view 
of  
the second line in \eqref{evs}, it must be $\mu_k \geq 0$, 
then $\mu_\alpha \geq 0$ for $\alpha > k$. This 
implies that $\sum_{\alpha=1}^{k'} \mu_\alpha \geq 0$. 
\item For a non flat vacuum spacetime $R_{ab}=0$ and, in 
view of \eqref{zero} and the first 
line in \eqref{evs}, it must be 
$\lambda_1 < 0 <\lambda_{d-2}$, then \eqref{evs} and \eqref{zero} imply 
$\sum_{\alpha=1}^{k<d-2} \lambda_\alpha <0$.
\item This follows again from \eqref{evs}, \eqref{zero}
and the first 
line in \eqref{evs}.
\end{enumerate}
\end{proof}

\noindent
\textbf{Example:} 
FLRW cosmologies have zero Weyl tensor, so,
in view of Proposition \ref{Pgcc}.ii, 
those for which the NEC holds satisfy  the 
$k-$GCC 
for any $k$. \\ 

The $k-$GCC, being the equivalent for higher codimension  CTMs 
of what the NEC is for CTSs, guarantees that any $k-$CTM is 
eligible in Proposition \ref{focal} and Theorem \ref{sing-thm}. 
Its failure forces us to restrict ourselves 
to CTMs which are suitable 
oriented to satisfy \eqref{cc}. The existence of such CTMs is
sufficient 
to prove null geodesic incompleteness (a single such 
geodesics is all we need). Proposition \ref{focal} is 
used also  in the standard proof that a CTS cannot 
intersect the DOC if the NEC holds. 
As explained above, if the $k-$GCC, the generalization of the 
NEC,  
does not hold, as happens in general according to Proposition \ref{Pgcc}, we cannot use this type of arguments to 
prove that $k-$CTMs do not intersect the DOC.

\section{The curvature condition for specific CTMs} \label{specific}
In this section  we exhibit examples of 
spacetimes for 
which, despite the fact that the $k-$GCC is not satisfied, 
one can still find $k-$CTMs 
oriented in such ways that  condition \eqref{cc} required 
by 
the singularity Theorem \ref{sing-thm} holds. 
The following definition helps to formalize the idea 
of ``suitable orientation'' (compare with Definition \ref{Dgcc}):

\begin{dfn} \label{k-cc} A $k$ dimensional spacelike 
subspace $W$ of $T_pM$  satisfies the 
 $k-$curvature condition (k-CC) if,  for 
any null vector $n^a$ orthogonal to $W$ 
the condition
\begin{equation}\label{kcc}
R_{abcd}N^a N^c h_W^{bd} \geq 0
\end{equation}
holds, where  $N^a$ is the tangent to the geodesic $\gamma$ 
with initial condition $n^a$ and ${h_W}_{ab}$ is the parallel 
transport along $\gamma$  of the induced metric on $W$.
\end{dfn}

A $k-$CTM $S$ such that for every $p \in S$, 
$T_pS$ satisfies the $k-$CC is then suitable for the singularity theorem. 
The use of this definition promises to be unmanageable 
on  generic spacetimes but, as we show next, there are 
physically interesting examples of application.

\subsection{Warped spacetimes} \label{warped}

 Assume that the  spacetime
belongs to the large and rich  class of \textit{warped product} manifolds  $B \times_{r^2} F$: 
\begin{equation}
\text{d}s^2=\tilde{g}_{\a \b}(x) \text{d}x^\a\text{d}x^\b+r^2(x)\,\bar{g}_{AB}(\theta)\text{d}\theta^A\text{d}\theta^B.
\label{eq:warped metric}
\end{equation}
Here $\{x^1,x^2, \dots \}$ are local coordinates of the 
\textit{base manifold} $(B, \tilde g)$, 
which is Lorentzian,  
and $\{\theta^1, \theta^2,\dots\}$ are coordinates of the 
Riemannian \textit{fiber manifold} $(F, \bar g)$. 
We assume $r: B \to \mathbb{R}$ is positive definite.
The superscripts $\sim$ and 
$-$  will be attached to tensors belonging 
to the base and fiber respectively, and $\tilde g$ 
($\bar g$) and their inverses will be used to lower and raise indexes of tensors on $B$ and $F$ respectively.
As an example, for a spacetime vector $N^a=(N^\a, N^A)$ 
we find $N_a=(N_\a,r^2 N_A)$, then $N^a N_a=N^\a N_\a+ 
r^2 N^AN_A$. 
The submanifolds $\{ x \} \times F$ and $B \times 
\{ \theta \}$ of $B \times_{r^2} F$ are respectively called 
the \textit{fiber} and \textit{leaf}
through $(x,\theta) \in B \times_{r^2} F$.
Vectors of the form 
$(N^\a,0)$ [$(0,N^A)$] are tangent to leaves [fibers] 
and will be called horizontal 
[vertical].
A calculation shows (see, e.g. \cite{senov-warped} 
equations (A.3)-(A.4))  that, for \eqref {eq:warped metric}, 
the nonzero Christoffel symbols (mod symmetries) are 
\begin{equation} \label{cristo}
\Gamma^\a_{\b \g}= \tilde \Gamma^\a_{\b \g}, \;\; \; \;\Gamma^A_{\b C}= \frac{\p_\b r}{r} \delta^A_C, \;
\; \;\; \Gamma^\a_{BC}= -r (\tilde g^{\a \d}\p_\d r) \;\bar g_{BC}, \;\;\;\
\Gamma^A_{BC} = \bar \Gamma^A_{BC}.
\end{equation}
From this equation follows that 
the parallel transport along a
geodesic $\gamma$ of a vector
that is initially vertical and orthogonal to $\gamma$  
remains vertical along $\gamma$ (see also Appendix A 
in \cite{senov-warped} and
chapter 7 in \cite{oniell}). As a consequence, 
if a CTM $S$ is oriented such that 
$T_p S$ is vertical at every $p \in S$, then, 
at every point of a geodesic $\gamma$ orthogonal 
to $T_pS$ , $h_W$ in \eqref{kcc} will project onto a 
vertical 
subspace of the tangent space. Thus, if equation 
\eqref{kcc} holds 
for vertical subspaces of $T_qM$ at any $q$, those 
$k-$CTMs which are 
tangent at every point to vertical subspaces will be 
suitable to be used in Theorem \ref{sing-thm} as $T_pS$ 
will satisfy the $k-$CC in Definition \ref{k-cc} at every $p 
\in S$. Now, since the non trivial (mod symmetries)  
components of the Riemann tensor 
are (see, e.g., \cite{senov-warped}, 
equations (3.3)-(3.6))
\begin{equation} \begin{aligned}
R_{\a\b\g\d} &= \tilde{R}_{\a\b\g\d}, \\
R_{\a B \g D} &= -r \left( \tilde{\nabla}_\g \tilde \nabla_\a r \right) 
\bar{g}_{BD}, \\
R_{ABCD} &= r^2 \bar{R}_{ABCD} - r^2 |\tilde \nabla r|^2 \left( \bar{g}_{AC}\bar{g}_{BD} - \bar{g}_{AD}\bar{g}_{BC} \right),
\end{aligned} \end{equation}
we find that, 
for $N^a=(N^\a, N^A)$ and 
$e^A_i, i=1,2,..,k$ 
an orthonormal basis of $W$, condition \eqref{kcc}  
reduces to 
\begin{equation}
\bar{R}_{ABCD}\bar{N}^A\bar{N}^Ce^B_ie^D_j\delta^{ij}-k\left(\bar{N}^A\bar{N}_A \,\tilde \nabla^\b r \tilde 
\nabla_\b 
r + \tilde{N}^\a\tilde{N}^\b \frac{\tilde 
\nabla_\a \tilde \nabla_\b r}{r}\right) \geq 0.
\label{eq:gcc warped}
\end{equation}
Since the fiber metric is positive definite and $N^a$ is null we have  $\bar{N}^A\bar{N}_A=-r^{-2}\,\tilde{N^\a}\tilde{N_\a}\geq 0$. 
The first term in \eqref{eq:gcc warped} is the sum of $k$ 
\textit{sectional curvatures} of the fiber; this  is a 
Riemannian submanifold, so the planes spanned by $\bar{N}$ and 
$e_i$ are non-degenerate. \\
Whenever  $dr \neq 0$,  $r$ can be taken as one 
of the $x^\a$ coordinates, in which case 
\begin{equation}
\begin{aligned}
&\tilde  \nabla^\b r \tilde \nabla_\b r \equiv \tilde{g}^{rr},\\
& \tilde{N}^\a\tilde{N}^\b \tilde \nabla_\a \tilde 
\nabla_\b r \equiv -\tilde{N}^\a\tilde{N}^\b\, \tilde{\Gamma}^r_{\a\b},
\end{aligned}
\end{equation}
and 
\eqref{eq:gcc warped} reads
\begin{equation}
\bar{R}_{ABCD}\bar{N}^A\bar{N}^Ce^B_ie^D_j\delta^{ij} 
- k \left(\tilde g^{rr} \bar N^A \bar N_A -\tilde{N}^\a\tilde{N}^\b\,\frac{ \tilde{\Gamma}^r_{\a\b}}{r} \right)\geq 0.
\label{eq: gcc warped simp0}
\end{equation}

\noindent
The results above can be gathered in the following:

\begin{prop} \label{gctms} For the warped spacetime 
\eqref{eq:warped metric}, if   a $k-$CTM $S$ has $T_pS$ 
tangent to the 
fibers for every $p$, the parallel transport 
of the tangent spaces along future normal null geodesics 
remain tangent to the fibers 
and $T_pS$ satisfies the $k-$cc in Definition \ref{k-cc} 
iff 
\eqref{eq:gcc warped} (equivalently \eqref{eq: gcc warped 
simp0}) holds along  future null geodesics normal to $T_pS$.
\end{prop}

In the remaining of this Section we will restrict to CTMs  
satisfying the hypothesis of Proposition \ref{gctms} (which will be called \textit{properly oriented} from now on).

\subsection{Warped spacetimes with constant curvature fibers}

For warped spacetimes, further simplifications arise  in 
those cases where the fiber has constant curvature $\bar{C}$,
\begin{equation}\label{ccw}
\bar{R}_{ABCD}=\bar{C}(\bar{g}_{AC}\bar{g}_{DB}-\bar{g}_{AD}\bar{g}_{CB}),
\end{equation}
 as the contraction in \eqref{eq:gcc warped} 
 reduces to 
\begin{equation}\label{simp}
\bar{R}_{ABCD}\bar{N}^A\bar{N}^Ce^B_ie^D_j\delta^{ij}=\bar{C}k\,\bar{N}^A\bar{N}_A. 
\end{equation}
and  \eqref{eq: gcc warped simp0} simplifies to
\begin{equation} \label{eq: gcc warped simp1}
\begin{split}
0 &\leq \left(\bar{C}-\tilde{g}^{rr}\right) \,\bar{N}^A\bar{N}_A  + \tilde{N}^\a\tilde{N}^\b\,\frac{ \tilde{\Gamma}^r_{\a\b}}{r} \\
&= \tilde N^\a \tilde N^\b \left[ \left(
\tilde{g}^{rr}-\bar{C}\right) \,\frac{\tilde g_{\a \b}}{r^2} +
\frac{ \tilde{\Gamma}^r_{\a\b}}{r} \right]
\end{split}
\end{equation}
which, interestingly, \textit{does not} depend \emph{explicitly} on $k$, although the dimension of the 
trapped submanifold matters. For example, if $k=d-2$ the 
first term in the first line of \eqref{eq: gcc warped simp0}
is absent since $\tilde 
N$ is a null vector of the Lorentzian base manifold 
(for higher codimension this vector is timelike in general). 
In view of Proposition \ref{gctms}, if a CTM is tangent 
to the fibers at every point, it will satisfy the hypothesis of Theorem \ref{sing-thm} iff \eqref{eq: gcc warped simp1} holds 
along its future null orthogonal geodesics. 
Some examples follow.

\subsection{Spherical symmetry in arbitrary dimensions}

In $d$ dimensions, in those regions where the area radius $r$ can be used as 
a coordinate ($dr \neq 0$), the most general spherically 
symmetric metric can be written in the form
\begin{equation}
    \text{d}s^2= -e^{2\beta}f\text{d}v^2+2e^{\beta}\text{d}v\text{d}r+r^2\gamma_{AB}\text{d}\theta^A\text{d}\theta^B,
\label{eq: vaidya metric}
\end{equation}
where  $\beta\equiv \beta(v,r)$, $f\equiv f(v,r)=1-
\frac{\mu(v,r)}{r^{d-3}}$, with  $\mu(v,r)$ the 
\textit{Misner-Sharp mass function} and 
$\gamma_{AB}\text{d}\theta^A\text{d}\theta^B$ is the 
metric 
of the unit $(d-2)$ sphere ${\rm S}^{d-2}$, 
which, in 
hyper-spherical coordinates $\theta^A$, reads
\begin{equation} \label{esferas}
\gamma_{AB}\; \text{d}\theta^A\text{d}\theta^B =
\sum_{B=1}^{d-2}  C_B\, \left(\text{d} \theta^B\right)^2,
\end{equation}
with
\begin{equation} \label{CB}
C_B = \begin{cases} 1 &, B=1 \\ 
\prod\limits_{C=1}^{B-1} \sin^2 (\theta^C) &, B>1 \end{cases}
\end{equation}
 The metric \eqref{eq: vaidya metric} can 
be used to model 
spherical collapse as well as static spherical 
symmetric black holes. 
If the 
spacetime is to satisfy the NEC and, in the dynamical case, 
we assume that the mass-energy is ingoing, then it is straightforward  to see that the following conditions must hold
\begin{subequations}
\begin{align}
&\partial_r\beta \geq 0, \label{condiciones sc 1}\\ 
&\partial_r\mu \geq r^{d-3}f  \partial_r\beta, \label{condiciones sc 2} \\
&\partial_v\mu\geq 0. \label{condiciones sc 3}
\end{align}   
\end{subequations}
Equation  \eqref{eq: vaidya metric} gives 
\begin{equation}
\begin{aligned}
\tilde{N}^a\tilde{N}^b\, \tilde{\Gamma}^r_{ab}&=-\left(f\partial_r\beta + \frac{\partial_rf}{2}\right)\tilde{N}^a\tilde{N}_a + \left(\tilde{N}^r\right)^2 \partial_r\beta - \left(\tilde{N}^v\right)^2 \partial_vf \\
&=\left(r^2f\partial_r\beta -\frac{\partial_r\mu}{r^{d-5}}+(d-3)\frac{\mu}{r^{d-4}} \right)\bar{N}^A\bar{N}_A + \left(\tilde{N}^r\right)^2 \partial_r\beta + \left(\tilde{N}^v\right)^2 \frac{\partial_v\mu}{r},
\end{aligned}   
\end{equation}
so the curvature condition \eqref{eq: gcc warped simp1} is equivalent to 
\begin{equation}
\left[ (d-1)\mu - r\,\partial_r \mu + r^{d-2}f\partial_r\beta \right]\bar{N}^A\bar{N}_A + \left(\tilde{N}^r\right)^2 r^{d-4}\partial_r\beta + \left(\tilde{N}^v\right)^2 r^{d-5}\partial_v\mu \geq 0.
\label{eq: gcc sc}
\end{equation}
In virtue of \eqref{condiciones sc 1} and \eqref{condiciones sc 3}, the second and third terms in \eqref{eq: gcc sc} are always non-negative; any negative contribution must come from the first term. Indeed, because of \eqref{condiciones sc 2}, it is easy to see that $- r\,\partial_r \mu + r^{d-2}f\partial_r\beta\leq 0$. \\


As a case of interest we search for CTMs within  the succession of submanifolds given by the $k-$spheres 
${\rm S}_o^{k}$, 
$k=d-2-n$, which are  defined by fixing $v=v_o$, $r=r_o$ and  $\theta^A=\theta^A_o$ for $A=1, 2,...,n$.   Introduce then the notation 
$\theta_o=(\theta^1_o, \theta^2_o,...,\theta^n_o)$. 
Note from \eqref{eq: vaidya metric}-\eqref{CB}
that the 
induced metric on these spheres is 
\begin{equation} \label{esferas2}
\text{d} s_{(k,o)}^2 =
r_o^2 \sum_{B=n+1}^{d-2}  
C_B\;  \left(\text{d} \theta^B \right)^2 = r_o^2 C_{n+1}(\theta_o) \, 
\text{d}s_{(k)}^2
\end{equation}
where $\text{d}s_{(k)}^2=  
\gamma^{(k)}_{I' J'} \text{d}\theta^{I'} \text{d} \theta^{J'}$ with
$I',J'=n+1,...,d-2$ is the metric on the  
$k=d-2-n$ dimensional unit sphere ${\rm S}^k$. 
The  $k-$volume of ${\rm S}^k_o$ is 
\begin{equation}\label{ksv}
\text{Vol}({\rm S}^k_{o}(\theta_o)) = r_o^k \, (C_{n+1}(\theta_o))^{k/2}  \;\text{Vol}({\rm S}^k).
\end{equation}
Since the coordinates 
\eqref{eq: vaidya metric}-\eqref{esferas} are 
adapted to ${\rm S}^k_o$ and $\p_v, \p_r, 
\p_{\theta^1},..., \p_{\theta^n}$ and orthogonal to it,
the mean curvature vector field can be easily calculated as the trace over $T {\rm S}^k_o$ of 
(minus) 
 the relevant Christoffel symbols: 
\begin{equation}
-\Gamma_{IJ}^{\phantom{I}a}\big|_{{\rm S}^k_o}=\frac{1}{2}g^{a b}\partial_{b}g_{IJ}=  \begin{cases}r_0\, C_{n+1}(\theta_o)\,g^{a r}(v_o,r_o)\tilde{\g}_{I J}(\theta) & a=v,r \\ \frac{C_{n+1}(\theta_o)}{C_{A}(\theta_o)}\, \text{cot}(\theta^A_o) \tilde{\g}_{IJ}(\theta) & a=A=1,2,...,n
\end{cases},
\label{eq: chr de ST}
\end{equation}
where the middle of the alphabet capital indices $I,J$ 
correspond to the submanifold angles, that is, 
they run from $n+1$ to $d-2$, and the second line 
is absent if $n=0$.
Taking the $T{\rm S}_o^k$-trace gives the mean curvature vector 
field 
\begin{equation}
H=\frac{k}{r_0}\left(e^{-\b}\partial_v + f\partial_r + {{\hat \sum}{}}_{A=1}^{n} \ \frac{1}{r_0C_{A}(\theta_o)}\text{cot}(\theta^A_o)\partial_{\theta^A} \right),
\end{equation}
where the hat indicates the sum is absent if $n=0$.
Requiring that $H^a$ be timelike 
gives the trapping condition for $S_o^k$:
\begin{equation}
H^a H_a = \frac{k^2}{r_o^2} \left[f + \hat\sum_{A=1}^n \frac{1}{C_{A}}\text{cot}^2(\theta_{0}^A) \right]
< 0.
\label{eq: trapping of S^k}
\end{equation}
The term in square brackets simplifies  to
\begin{equation} \label{trapped}
f + \hat \sum_{A=1}^n 
\frac{1}{C_{A}}\text{cot}^2(\theta_{0}^A) = 
\begin{cases} 
f &, \text{if} \; n=0 \\
f -1 + \frac{1}{C_{n+1}(\theta_o)} &, \text{if} \, n \geq 1
\end{cases}
\end{equation}
The case $n=0$ above recalls us that the $d-2$ spheres are trapped 
if they are inside the apparent horizon \cite{senovsc}, the $n \geq 1$ 
cases set a minimum $k-$volume for ${\rm S}^k_o$
(see \eqref{ksv}, note that $\mu$ has dimension 
$r^{d-3}$):
\begin{equation}\label{cond}
\frac{1}{C_{n+1}(\theta_o)} < \frac{\mu(v_o,r_o)}{r_o^{d-3}},
\end{equation}
or, equivalently 
\begin{equation}\label{cond1}
\left[\frac{\text{Vol}({\rm S}^k_o)}{ r_o^k  \, \text{Vol}({\rm S}^k)} \right]^{2/k}
>  \frac{r_o^{d-3}}{\mu(v_o,r_o)}
\end{equation}
Note that, since $C_{n+1}(\theta_o) \leq 1$, inequality 
\eqref{cond}  implies $f \leq 0$; thus,  our
trapped $k-$spheres can be found only within the 
apparent horizon. Moreover, as the supporting 
trapped $d-2$ sphere approaches the apparent horizon, the only contained lower-dimensional trapped spheres  fit  around  the generalized 
great circles: $\theta^1_o=\theta^2_o=...=\theta^n_o=\pi/2$.\\

We can make a more geometric description of what we have found, 
independent  of any specific set of angular coordinates / axis choice: 
the relation of the hyper-spherical coordinates $\theta^A$ of a $d-2$ dimensional  sphere with 
the Cartesian coordinates of its Euclidean ambient space 
$\mathbb{R}^{d-1}$ is
\begin{equation}
\begin{split}
  &x^1 = r \cos(\theta^1), \\
  &x^2 = r \sin(\theta^1) \cos(\theta^2), \\
  &x^3 = r \sin(\theta^1) \sin(\theta^2) \cos(\theta^3), \\
      &\qquad \vdots\\
  &x^{d-2} = r \sin(\theta^1) \cdots \sin(\theta^{d-3}) \cos(\theta^{d-2}), \\
  &x^{d-1}     = r \sin(\theta^1) \cdots \sin(\theta^{d-3}) \sin(\theta^{d-2}).
\end{split}
\end{equation}
Fixing 
$r=r_o$,  $\theta^1=\theta^1_o$, ... ,$\theta^n=\theta^n_o$ 
 defines an affine space $\mathcal{A} \subset \mathbb{R}^{d-
 2}$:   $x^A=x^A_o$ for $A=1,2,...,n$, leaving a cut of the $d-2$ dimensional sphere of radius $r_o$ 
with $\mathcal{A}$. The resulting $k$-sphere  
has the $SO(k)$ isometry subgroup $H$ 
 consisting of matrices of the form 
$g= \text{diag}( \mathbb{I}_{n \times n}, h)$,
$h \in SO(k)$. 
\emph{Any other} conjugate $SO(k)$ subgroup of $SO(d-2)$, 
say $gHg^{-1}$,  
would do the same: its orbit (now the intersection 
of the $(d-2)$-sphere with $g \mathcal{A}$)
would be a $k-$sphere, 
and this will be 
trapped iff \eqref{cond1} holds. As mentioned before, as the supporting sphere approaches the horizon, the $k-$dimensional trapped spheres 
will tight around  the intersection with 
subspaces of $\mathbb{R}^{d-1}$, as, according to \eqref{cond},
$x_o^A \to 0$ for 
$A=1,2,...,n$.\\

Having found a rich set of trapped submanifolds of dimensions 
$k=1,2,...,d-2$ in the most general spherically symmetric spacetime, 
we would like to know if these are 
singularity predictors. According to Proposition \ref{gctms} 
and using the simplifications found for   spherically symmetry, 
this will be the case if: i) the (proper extension of) the 
metric \eqref{eq: vaidya metric} admits 
a non compact Cauchy surface and ii) equation \eqref{eq: gcc sc} holds 
along future null normal geodesics of the CTM. 
Analytic extensions are handled 
the same way as in $3+1$, as they involve only the base space   and so are independent of the value 
of $d \geq 4$. Condition (ii) needs to be checked case by case. 
Some examples follow.

\subsubsection{Schwarzschild-Tangherlini BH} Inserting $\beta(v,r)=0$ and $\mu(v,r)=\mu>0$ in \eqref{eq: vaidya metric} gives the generalization of the Schwarzschild BH to arbitrary dimensions \cite{reall}. In this case it is trivial to see that \eqref{eq: gcc sc} is fulfilled {\em at every point of the spacetime}, in particular, along the future null normal geodesics of 
our trapped $k-$spheres. As mentioned above, by Proposition \ref{gctms}, every properly oriented CTM (for any $k$) will satisfy the hypotheses of Theorem \ref{sing-thm} and thus act as a singularity predictor.

\subsubsection{Vaidya Spacetime} The Vaidya spacetime of incoming radiation is trivially extended to arbitrary dimensions by replacing $\beta(v,r)=0$ and $\mu(v,r)=\mu(v)>0$ (with $d \mu/dv \geq 0$) in \eqref{eq: vaidya metric}. As there is no dependence in $r$, the only negative contribution in \eqref{eq: gcc sc} vanishes, so the  curvature condition is satisfied. The conclusion is same as in the previous example.

\subsubsection{Reissner-Nördstrom BH}

A more intriguing case is the spherically symmetric charged sub-extreme black hole, whose metric depends on two parameters: the mass $m>0$ and the electric charge $q\neq 0$, subject to $m\geq |q|$. For the sake of simplicity we will fix $d=4$, then the metric is \eqref{eq: vaidya metric} with $\beta(v,r)=0$ and $f(v,r)=1-\frac{2m}{r}+\frac{q^2}{r^2}$. The zeroes of 
$f$ are the Cauchy ($r_-$) and event ($r_+$) horizon radii
\begin{equation}\label{hor}
r_\pm = m\pm \sqrt{m^2-q^2}.
\end{equation}
The maximal analytic extension contains infinitely many copies of regions I ($r>r_+$), II ($r_- < r<r_+$) and III ($0<r<r_-$). 
We will consider instead 
the globally hyperbolic spacetime $M^{3+1}_{RN}$,  
depicted in Figure \ref{fig}, which consists of two copies of region II and two copies of region I (primes are used in the figure to identify isometric regions). Region III is physically 
irrelevant: this static solution of the Einstein-Maxwell equations 
is linearly unstable: perturbations inside it grow exponentially in time \cite{Dotti0}. As it is well known 
there are infinitely many possible (non analytic) 
extensions of $M^{3+1}_{RN}$
beyond the Cauchy (future and past) horizons (see, e.g., \cite{hawk}, page 159).  Moreover, these  
horizons  are unstable and become  a curvature singularity if we replace the $M^{3+1}_{RN}$ data on a Cauchy surface, such as  $\Sigma$ in Figure \ref{fig}, with  data close to it \cite{wald,simpson, Poisson,dafermos}.

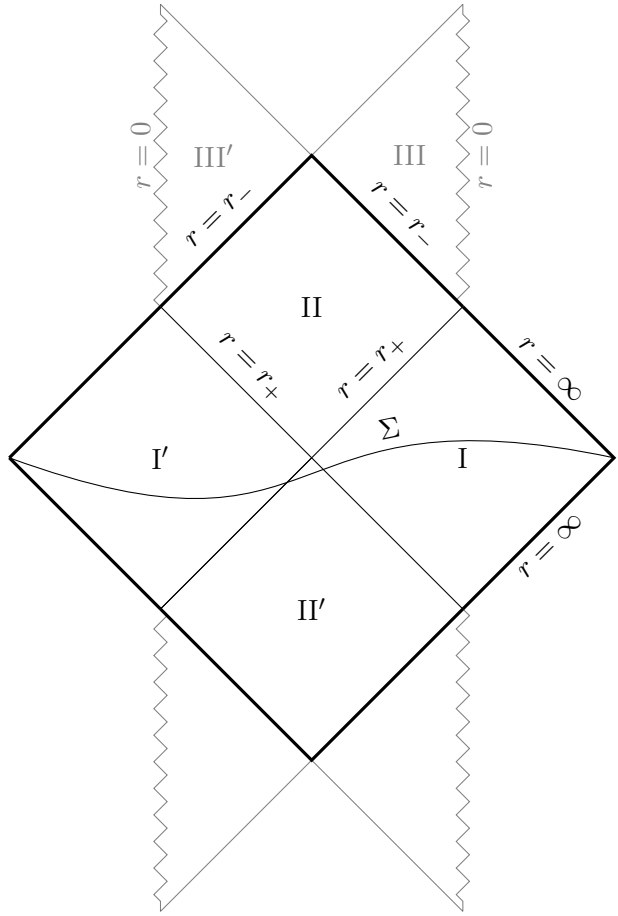
\begin{figure}[h]
    \centering
\begin{tikzpicture}
		\node (I)    at ( 2,0)   {I};
		\node (I$'$)   at (-2,0)   {I$'$};
		\node (II)  at (0, 2) {II};
		\node (II$'$)   at (0,-2) {II$'$};
		\node (III)  at (2, 4) [label={[xshift=-20pt,opacity=0.5]center:III}]  {};
		\node (III$'$)   at (-2,4) [label={[xshift=20pt,opacity=0.5]center:III$'$}] {};
		
		\path  
		(I$'$) +(90:2)  coordinate  (I$'$top)
		+(-90:2) coordinate (I$'$bot)
		+(0:2)   coordinate                  (I$'$right)
		+(180:2) coordinate (I$'$left)
		;
		\draw [very thick] (I$'$left) -- (I$'$top);
		\draw(I$'$right) -- (I$'$bot);
		\draw [very thick] (I$'$bot) -- (I$'$left);
		
		\path 
		(I) +(90:2) coordinate
		 (Itop)
		+(-90:2) coordinate
		 (Ibot)
		+(180:2) coordinate (Ileft)
		+(0:2)   coordinate
		(Iright)
		;		
		\draw  (Ileft) -- 
		node[midway, above, sloped]    {$r=r_+$}
		node[midway, below, sloped] {}
		(Itop);
		\draw [very thick] (Itop) --
		node[midway, above, sloped] {$r=\infty$}
		(Iright) -- 
		node[midway, below, sloped] {$r=\infty$}
		(Ibot);
		\draw (Ibot) -- (Ileft);
		
		\path  
		(II) +(90:2)  coordinate  (IItop)
		+(-90:2) coordinate (IIbot)
		+(0:2)   coordinate                  (IIright)
		+(180:2) coordinate (IIleft)
		;
		\draw [very thick] (IIleft) -- node[midway, above, sloped]    {$r=r_-$} (IItop) -- node[midway, above, sloped]    {$r=r_-$} (IIright); 
		\draw (IIbot) -- node[midway, above, sloped]    {$r=r_+$} (IIleft);
		
		\path  
		(II$'$) +(90:2)  coordinate  (II$'$top)
		+(-90:2) coordinate (II$'$bot)
		+(0:2)   coordinate                  (II$'$right)
		+(180:2) coordinate (II$'$left)
		;
		\draw (II$'$left) -- (II$'$top);
		
		\draw [very thick] (II$'$right) -- (II$'$bot) -- (II$'$left);
		
		\path  
		(III$'$) +(90:2)  coordinate  (III$'$top)
		+(-90:2) coordinate (III$'$bot)
		+(0:2)   coordinate                  (III$'$right)
		+(180:2) coordinate (III$'$left)
		;
		\draw [opacity=0.5](III$'$top) -- (III$'$right);
		
		\path  
		(III) +(90:2)  coordinate  (IIItop)
		+(-90:2) coordinate (IIIbot)
		+(0:2)   coordinate                  (IIIright)
		+(180:2) coordinate (IIIleft)
		;
		\draw [opacity=0.5] (IIIleft) -- (IIItop) ;
		
        \draw[decorate,decoration=zigzag,opacity=0.5] (I$'$top) -- (III$'$top)
		node[midway, above, sloped] {$r=0$};		
		\draw[decorate,decoration=zigzag,opacity=0.5] (Itop) -- (IIItop)
		node[midway, below, sloped] {$r=0$};		
		\draw[decorate,decoration=zigzag,opacity=0.5] (I$'$bot) -- (-2,-6);		
		\draw[decorate,decoration=zigzag,opacity=0.5] (Ibot) -- (2,-6);
		
		\draw[opacity=0.5] (II$'$bot) -- (-2,-6);
		
		\draw[opacity=0.5] (II$'$bot) -- (2,-6);
		
		\draw 
		(Iright) to[out=168, in=-20, looseness=1.5] node[pos=0.3, above, sloped] {$\Sigma$} (I$'$left) ;
		
	\end{tikzpicture}
    \caption{The spacetime $M^{3+1}_{RN}$ considered in the example is the globally hyperbolic open  subset  of the maximally analytic extension of the 
    Reissner-Nördstrom Einstein-Maxwell field equations  solution shown in this figure.
    $\Sigma$ is a Cauchy surface.}
    \label{fig}
    \end{figure}

It is well known that Penrose's singularity theorem cannot be 
used in any globally hyperbolic open subset of 
the maximal analytic extension 
to predict the geodesic incompleteness 
caused by the $r=0$ singularity, since it is beyond the Cauchy horizon (see \cite{minguzzi} for 
alternative singularity theorems that fit in this case). 
However the trapped 
spheres in region II of the globally hyperbolic space 
$M^{3+1}_{RN}$, which satisfies the NEC and admits a non-compact Cauchy surface (e.g., 
$\Sigma$ in Figure \ref{fig}), predict  the future null geodesic incompleteness 
caused by the Cauchy horizon at $r=r_-$. A slight variation of the initial data at $\Sigma$ would preserve theses trapped surfaces, which now would predict the 
singular hypersurfaces  where the perturbed spacetime ends. \\
A natural question to ask is if we could use  higher  codimension CTMs (that is TLs) of $M^{3+1}_{RN}$ to predict its incompleteness. Let us switch the notation for 
spherical coordinates from $(\theta^1,\theta^2)$ used above to the standard names $(\theta,\phi)$. Condition \eqref{cond} 
for a parallel $\theta=\theta_o$ to be a TL can be 
written, using \eqref{hor}, as 
\begin{equation}
\frac{(r_o-r_+)(r_o-r_-)}{{r_o}^2} + \cot^2 \theta_o <0.
\end{equation}
This equation implies that trapped parallels occur in region II (and its isometric copy 
II'), in a band of the sphere they belong to that is symmetric around the Equator 
$\theta=\pi/2$, and that this band 
 narrows to zero width as $r \to r_+$ from the left 
or $r \to r_-$ from the right. Can we use these TLs in Theorem \ref{sing-thm} to predict the geodesic incompleteness 
caused by the Cauchy horizon?
  Condition \eqref{eq: gcc sc} in this case results in
\begin{equation} \label{TLcond}
    \frac{(3mr-2q^2)}{r}\bar{N}^A\bar{N}_A \geq 0,
\end{equation}
\noindent which is fulfilled when $r\geq \frac{2q^2}{3m}$. 
A TL then satisfies the hypothesis of Theorem \ref{sing-thm} 
if \eqref{TLcond} holds all the way from the TL 
value $r=r_o$ (which lies in region II),  to the Cauchy horizon at $r=r_-$. This will be the case 
iff $3mr_--2q^2>0$ or, equivalently, if the condition 
\begin{equation}
\frac{|q|}{m} > \frac{\sqrt{3}}{2}\sim 0.87
\end{equation}
is added to the sub-extreme requirement that  $|q|<m$. Note 
that this is different from the case of CTSs, which 
predict geodesic incompleteness for arbitrary 
$|q| <m$.

\section{Conclusions}\label{conclusion}

Penrose's singularity theorem proves 
that the existence of
a CTS anticipates the incompleteness 
of a future null geodesic orthogonal to it. 
The theorem  requires that $R_{ab}N^a N^b \geq 0$ along every future null 
geodesics normal to the CTS ($N^a$ is   
the geodesic tangent vector), 
a condition that is guaranteed if the  
stronger NEC is satisfied. The 
generalization of Penroses's theorem 
to CTMs of higher codimension, given 
in \cite{Gallo}, 
requires that  the curvature 
condition \eqref{cc} holds along future null orthogonal geodesics. This  
 is  guaranteed if the $k-$GCC introduced in 
Definition \ref{Dgcc} (the analogous of the NEC), is satisfied. 
In this paper we prove a number of useful properties 
of the $k-$GCC, these are gathered in Proposition \ref{Pgcc}. 
We also  show  that 
the $k-$GCC is violated in generic 
spacetimes, in which case  not every CTM predicts geodesic incompleteness. However,  particular 
CTMs can be found  for which the weaker condition \eqref{cc} holds and  Theorem \ref{sing-thm} applies. Since these 
CTMs predict singularities at their future, in view of the weak cosmic censorship, they should  be regarded as 
signatures of black hole interiors.  A number 
of examples are worked out, including   CTMs 
of dimensions ranging from one to $d-1$ in $d+1$ dimensional 
spacetimes.  
The fact that the $k-$GCC does not hold in generic spacetimes, does not allow the use of Proposition \ref{focal} to prove, using standard arguments, that $k-$CTMs cannot intersect the DOC.

\end{document}